\documentclass[a4paper]{lipics-v2021}
\usepackage{mathtools}
\usepackage{cite}
\newcommand{\RE}{\mathbb{R}}

\newcommand{\bdOmega}{\partial \kern+1pt \Omega} 

\newcommand{\SP}{\kern+1pt}             
\newcommand{\revFunk}[1]{{}^r \kern-1pt F_{#1}}
\DeclareMathOperator{\interior}{int}

\title{Software for the Thompson and Funk Polygonal Geometry}
\titlerunning{Software for the Thompson and Funk Polygonal Geometry}

\author{Hridhaan Banerjee}{ Thomas Jefferson High School for Science and Technology, Virginia, USA \and \url{~}}{hridhaan.s.banerjee@gmail.com}{}{}

\author{Carmen Isabel Day}{California State University Channel Islands, California, USA  \and \url{~}}{carmen.isabel.day@gmail.com}{}{}

\author{Auguste H. Gezalyan}{Department of Computer Science, University of Maryland, College Park, USA \and \url{~}}{octavo@umd.edu}{https://orcid.org/0000-0002-5704-312X}{}

\author{Olga Golovatskaia}{Mount Holyoke College, Massachusetts, USA\and \url{~}}{olga.golovatskaia@gmail.com}{}{}

\author{Megan Hunleth}{Montgomery Blair High School, Silver Spring, Maryland, USA  \and \url{~}}{megan@hunleth.com}{}{}

\author{Sarah Hwang}{University of Maryland, College Park, USA \and \url{~}}{shwang18@terpmail.umd.edu}{}{}

\author{Nithin Parepally}{Department of Computer Science, University of Maryland, College Park, USA \and \url{~}}{nparepa@terpmail.umd.edu}{}{}

\author{Lucy Wang}{Princeton University, New Jersey, USA \and \url{~}}{lucywangj@gmail.com}{}{}

 \author{David M. Mount}{Department of Computer Science, University of Maryland, College Park, Maryland, USA \and \url{https://www.cs.umd.edu/~mount/}}{mount@umd.edu}{https://orcid.org/0000-0002-3290-8932}{}

\authorrunning{Banerjee, Day, Hunleth, Hwang, Gezalyan, Golovatskaia, Parepally, Wang, and Mount}

\Copyright{H. Banerjee, C. I. Day, M. Hunleth, S. Hwang, A. H. Gezalyan, O. Golovatskaia, N. Parepally, L. Wang, and D. Mount}


\date{\today}

\EventEditors{Wolfgang Mulzer and Jeff M. Phillips}
\EventNoEds{2}
\EventLongTitle{40th International Symposium on Computational Geometry (SoCG 2024)}
\EventShortTitle{SoCG 2024}
\EventAcronym{SoCG}
\EventYear{2024}
\EventDate{June 11-14, 2024}
\EventLocation{Athens, Greece}
\EventLogo{socg-logo.pdf}
\SeriesVolume{293}
\ArticleNo{88}

\hideLIPIcs
\nolinenumbers
\begin{document}

\maketitle

\begin{abstract}
Metric spaces defined within convex polygons, such as the Thompson, Funk, reverse Funk, and Hilbert metrics, are subjects of recent exploration and study in computational geometry. This paper contributes an educational piece of software for understanding these unique geometries while also providing a tool to support their research. We provide dynamic software for manipulating the Funk, reverse Funk, and Thompson balls in convex polygonal domains. Additionally, we provide a visualization program for traversing the Hilbert polygonal geometry.

\end{abstract}

\section{Introduction}
Thompson, Funk, reverse Funk, and Hilbert are closely related distance measures that apply to points in the interior of a convex body (see Section~\ref{sec:defs} for definitions). The Hilbert metric has applications in convex approximation \cite{abdelkader2018delone,abdelkader2024convex,vernicos2021flag}, real analysis \cite{lemmens2014birkhoff}, linear algebra \cite{liverani1994generalization}, clustering \cite{nielsen2019clustering}, and machine learning \cite{vaneceksupport} and has been the subject of recent study for this reason \cite{nielsen2017balls,gezalyan2023delaunay, gezalyan2023voronoi, bumpus2023software, banerjee2024heine}. The Funk, reverse Funk, and Thompson metrics are comparatively unexplored but all share similarities with the Hilbert metric \cite{thompson1963certain,papadopoulos2014handbook}. 

In this paper, we present new dynamic software for manipulating the Funk, reverse Funk, and Thompson balls in convex polygonal domains. Additionally, we provide visualization software for traversing the Hilbert polygonal geometry. Our code is predominantly written in Javascript and is available at \url{https://github.com/nithin1527/funk-geo-visualizer}. To use the software, please go to \url{https://funk-geo-visualizer.vercel.app/}. For more information on the usage, please see the README on the GitHub repository or watch the accompanying video available at \url{https://www.youtube.com/watch?v=IeRIO25iK-o}.

\section{Definitions} \label{sec:defs}

Throughout, $\Omega$ refers to a convex polygon in $\RE^2$ with $m$ sides. Let $\bdOmega$ denote its boundary. When clear, we omit explicit reference to $\Omega$.

\begin{definition}[Funk weak metric]
Given two distinct points $p,q \in  \interior \Omega$ in $\RE^d$, let $q'$ denote the intersection of the ray from $p$ through $q$ with $\bdOmega$. Define the \emph{Funk weak metric} to be:
\[
    F_\Omega(p,q)
         ~ = ~ \ln \frac{\|p - q'\|}{\|q - q'\|},
\]
where $F_\Omega(p,q)=0$.

\end{definition}

The above definition is also sometimes called the \emph{forward Funk metric}. Note that the Funk weak metric is asymmetric. Its reverse, the $\emph{reverse Funk metric}$, is defined to be $rF_\Omega(p,q)=F_\Omega(q,p)$. Define the \emph{Funk ball} centered at a point $p$ of radius $r$, denoted $B_{F}(p,r)$, to be the set of points $q \in \interior \Omega$ such that $F_{\Omega}(p,q) \leq r$. The reverse Funk ball is defined analogously. These are scaled homotheties of $\Omega$ (see Figure~\ref{fig:balls}(b) and (c)). 

\begin{lemma} \label{lem:ffunk-ball}
The open forward Funk ball of radius $r$ around a point $p$ is the image of $\Omega $ under Euclidean homothety
around the point $p$ with dilation factor of $(1-e^{-r})$\cite{papadopoulos2014handbook}.
\end{lemma}

\begin{lemma} \label{lem:rfunk-ball}
The open reverse Funk ball of radius $r$ around a point $p$ is the image of $\Omega $ under Euclidean homothety
around the point $p$ with dilation factor of $(e^{r}-1)$\cite{papadopoulos2014handbook}.
\end{lemma}

The \emph{Hilbert metric} can be defined as the average of the forward and reverse Funk metrics \cite{papadopoulos2014handbook}.

\begin{figure}[htbp]
    \centerline{\includegraphics[scale=0.6]{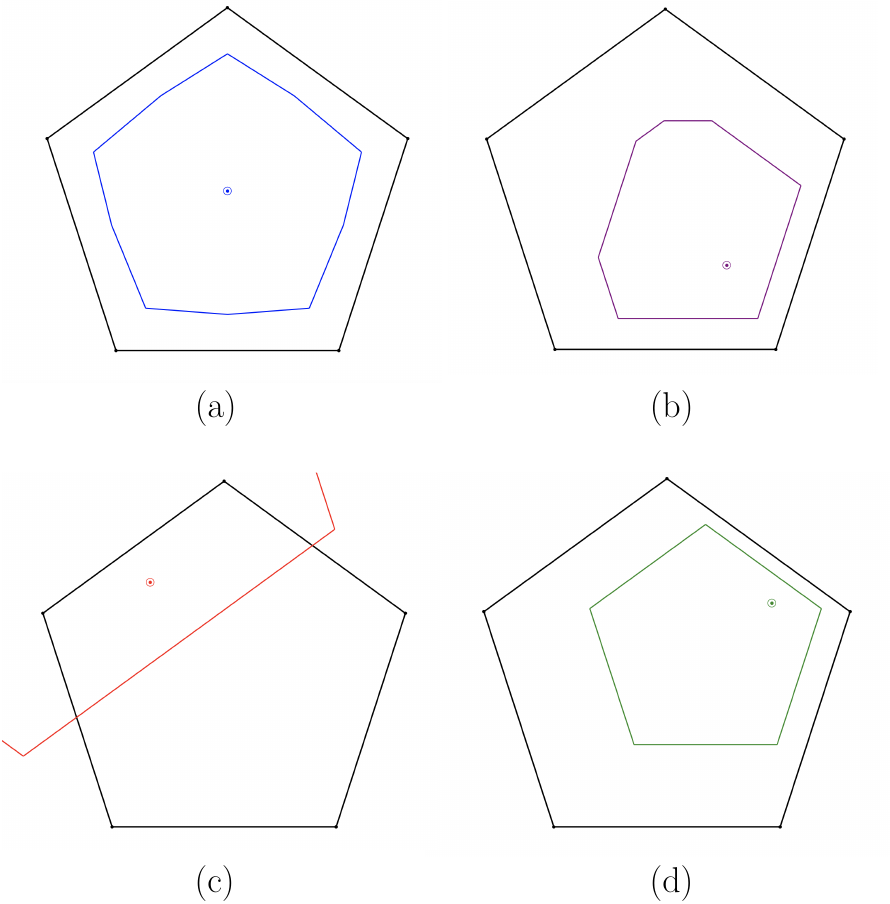}}
    \caption{The Hilbert (a), Funk (b), Reverse Funk (c), and Thompson (d) balls around a point} \label{fig:balls}
\end{figure}

\begin{definition}[Hilbert metric]
Given two distinct points $p,q$ in a convex polygonal $\Omega \subset \RE^d$ we define the \emph{Hilbert metric} to be:
\[
    H_\Omega(p,q)
         ~ = ~ \frac{1}{2}(F_\Omega(p,q) + rF_\Omega(p,q)),
\]
\end{definition}

Note that the Hilbert metric is a proper metric, obeying the identity of indiscernibles, symmetry, and the triangle inequality. Hilbert balls are polygons with at most $2m$ sides (see Figure~\ref{fig:balls}(a)). For a characterization of Hilbert balls see ``Balls in the Hilbert Polygonal Geometry'' by Nielsen and Shao \cite{nielsen2017balls}.

The notion of the Hilbert metric gave rise to the definition of another metric by A.C. Thompson in 1963 called the \emph{Thompson metric} \cite{thompson1963certain} on convex cones. The formulation of this metric can be generalized through the Funk metric.

\begin{definition}[Thompson metric]
Given two distinct points $p,q$ in a convex polygonal $\Omega\subset \RE^d$ we define the \emph{Thompson metric} to be:
\[
    T_\Omega(p,q)
         ~ = ~ \max (F_\Omega(p,q),rF_\Omega(p,q)),
\]
\end{definition}

Note that the identity of indiscernibles and triangle inequality clearly carry over from the Funk metric. The symmetry of the Thompson metric follows from the fact that it is the maximum of the forward and reverse of a weak metric. 

The Thompson ball (see Figure~\ref{fig:balls}(d)) is the intersection of the forward and reverse Funk balls. Since both these balls are convex (as they are homotheties of $\Omega$ intersected with $\Omega$\cite{papadopoulos2014handbook}), computing the Thompson ball reduces to calculating the intersection of two convex polygons. This can be done linearly in the complexity of the polygons. Hence, we get:

\begin{lemma} \label{lem:funk-ball-comp}
The Thompson ball around a point $p$ with radius $r$ has at most $2m$ sides and is the intersection of the forward and reverse Funk balls around $p$ with radius $r$.
\end{lemma}

Thompson balls have an interesting property that sets them apart from many other metrics. 

\begin{lemma}Thompson balls are not pseudo-disks (see Figure~\ref{fig:ThompsonAndThompson}(a)).\end{lemma}

\begin{lemma}A Thomspon ball of radius $r$ centered at $p$ can be nested between two Hilbert balls of radius $\frac{1}{2}r$ and $r$ centered at $p$ (see Figure~\ref{fig:ThompsonAndThompson}(b)).\end{lemma}

\begin{proof}
 Given a convex body $\Omega\subset \RE^d$, for every point $p \in \Omega$ and radius $0<r$, we will show:
\[
    B_H(p,r/2) ~ \subset ~ B_T(p,r) ~ \subset ~ B_H(p,r).
\]

 To prove that $B_T(p,r)\subset B_H(p,r)$, we must show that the Thompson distance from $p$ to a point on the boundary of the ball is greater than or equal to the Hilbert. This follows directly from the definition of the Hilbert and Thompson metrics: 
\[
    H(a,b)
        ~ = ~ \frac{1}{2}(F(a,b)+F(b,a))\leq 2 \cdot \frac{1}{2}(\max(F(a,b),F(b,a))=T(a,b). 
\]
To prove that $B_H(p,r/2) \subset B_T(p,r)$, we must show that the Thompson distance from $p$ to a point on the boundary of the ball is less than or equal to the Hilbert distance. This follows from the fact that if $H(a,b)=\frac{1}{2}r$ then $\frac{1}{2}(F(a,b)+F(b,a))=\frac{1}{2}r$ so $F(a,b)+F(b,a)=r$ and so $T(a,b)=\max(F(a,b),F(b,a))\leq r$. 
\end{proof}

\begin{figure}[htbp]
    \centerline{\includegraphics[scale=0.6]{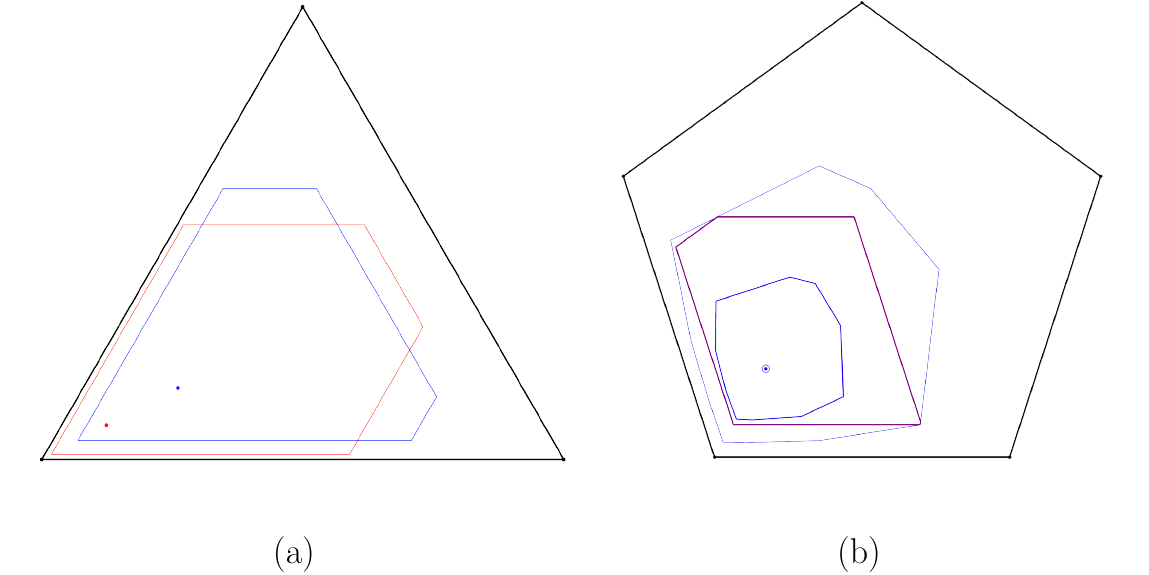}}
    \caption{(a) Thompson balls are not pseudo-disks, (b) a Thompson ball of radius $1$ between two Hilbert balls of radii $1/2$ and $1$. } \label{fig:ThompsonAndThompson}
\end{figure}

\section{Traversing the Hilbert Geometry}

\begin{figure}[htbp]
    \centerline{\includegraphics[scale=0.6]{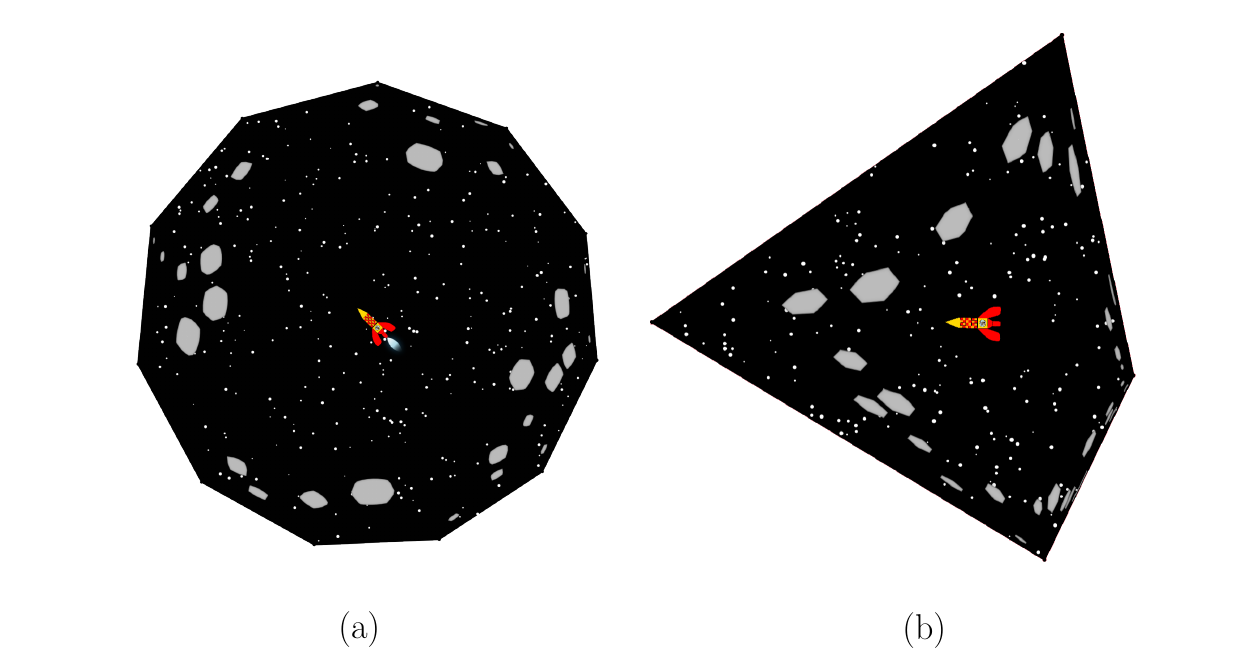}}
    \caption{(a) Travelling in a hendecagon and (b) after travelling some distance in the square. } \label{fig:HilbertTravel}
\end{figure}

We provide a visualization software for inserting Hilbert balls into a user-specified Hilbert polygonal geometry, and then traveling around that geometry (see Figure~\ref{fig:HilbertTravel}). In order to give users the ability to move in the Hilbert geometry, we implement a projective mapping that first affinely shifts $\Omega$ so that its centroid is at the origin $O$. Then, using a map $\phi_v$, we map $O+v$ to the centroid $\phi_v(\Omega)$ where $v$ is an accumulated displacement vector. This mapping is given by:
\[
    \phi_v(p)
         ~ = ~ \frac{p}{1+\langle p,v\rangle},
\]
where p is an arbitrary point\cite{izmestiev2023matching}.

However, this projective mapping often deforms $\Omega$ into very skinny forms. To handle this situation, we first capture the approximate John ellipsoid of the original shape scaled with the Mahalanobis distance. After movement, we calculate the new approximate John ellipsoid with respect to the projection of the vertices of $\Omega$. Then, we map the new John ellipsoid to the unit circle using a Cholesky decomposition. Lastly, we map the unit circle to the original approximate John ellipsoid to center the new geometry. Since all the maps we use are affine maps, they preserve the Hilbert metric. For more information on the mapping used, see the README on the GitHub repository.

\bibliography{shortcuts,hilbert}

\end{document}